%% file: main.tex
\documentclass[11pt]{article}
\usepackage{comment}
\usepackage[utf8]{inputenc}
\usepackage{fullpage}
\usepackage{amsmath}
\usepackage{amssymb}
\usepackage{amsthm}
\usepackage{mathtools}
\usepackage{hyperref}
\usepackage[linesnumbered, ruled]{algorithm2e}
\SetAlCapSkip{0.5em}
\usepackage[inline]{enumitem}
\usepackage{xcolor}
\usepackage{array}
\usepackage{float}
\usepackage{xfrac}
\usepackage[normalem]{ulem}
\usepackage{todonotes}
\usepackage{bbm}

\newtheorem{definition}{Definition}
\newtheorem{theorem}{Theorem}
\newtheorem{lemma}{Lemma}

\newtheorem{observation}{Observation}

\title{\Large Improved approximation algorithms for two Euclidean $k$-Center variants%
\thanks{The results of this work first appeared in the MSc thesis of the third author, where the first two authors served as co-advisors. An independent recent work of Lee, Nagarajan and Wang~\cite{new_paper} obtained a $(1+\sqrt{3})$-approximation algorithm for Robust Euclidean $k$-Supplier by using similar, but not identical, ideas.
}}

\author{Haris Angelidakis\thanks{Gnosis Ltd. Email: harisangelidakis@gmail.com. Research was conducted while the author was at ETH Zurich.} \and Ivan Sergeev\thanks{ETH Zurich. Email: isergeev@ethz.ch.} \and Pontus Westermark\thanks{Nozomi Networks. Email: pontusw25@gmail.com. Research was conducted while the author was at ETH Zurich.}}

\date{}

\begin{document}

\maketitle

\input{./abstract}

\input{./introduction}

\input{./hockbaum_shmoys_framework.tex}

\input{./matroid_center.tex}

\input{./k_supplier.tex}

\paragraph{Acknowledgments.} All three authors would like to thank Rico Zenklusen for helpful discussions and for suggesting to look for the deterministic reduction of Section~\ref{sec:matching} (rather than a randomized one that was using exact weight matching machinery and which appeared in a preliminary draft of this work). The first author would also like to thank Georg Anegg and Adam Kurpisz for useful discussions in the early stages of this work.

\bibliographystyle{plain}
\bibliography{references}

\end{document}

%% file: abstract.tex
\begin{abstract}

The $k$-Center problem is one of the most popular clustering problems. After decades of work, the complexity of most of its variants on general metrics is now well understood. Surprisingly, this is not the case for a natural setting that often arises in practice, namely the Euclidean setting, in which the input points are points in $\mathbb{R}^d$, and the distance between them is the standard $\ell_2$ Euclidean distance. In this work, we study two Euclidean $k$-Center variants, the Matroid Center problem on the real line and the Robust Euclidean $k$-Supplier problem, and provide algorithms that improve upon the best approximation guarantees known for these problems. 

The Matroid Center problem on the real line is one of the rare instances of a 1-dimensional $k$-Center variant that is NP-hard, as shown by Chen, Li, Liang, and Wang (2016); most $k$-Center problems become easy when restricted to the real line. In fact, Chen et al.~showed that the problem is $(2-\varepsilon)$-hard to approximate. On the algorithmic side, only the $3$-approximation algorithm for Matroid Center on general metrics by Chen et al.~is known for tackling the problem. In this work, building on the classic threshold technique of Hochbaum and Shmoys (1986) and by exploiting the very special structure of real-line metrics, we improve upon the $3$-approximation factor and provide a simple $2.5$-approximation algorithm.

We then turn to the Robust $k$-Supplier problem (also known as $k$-Supplier with outliers), which is one of the most popular $k$-Center variants that have been studied in the literature. It is known that the problem admits a $3$-approximation on general metrics, which is tight even when there are no outliers, assuming $\mathtt{P} \neq \mathtt{NP}$. We focus on the Euclidean setting, for which the $3-\varepsilon$ hardness does not hold anymore. For the special case where there are no outliers, Nagarajan, Schieber and Shachnai (2020) gave a very elegant $(1+\sqrt{3})$-approximation algorithm for the Euclidean $k$-Supplier problem, thus overcoming the $3-\varepsilon$ barrier. However, their algorithm does not generalize to the robust setting. In this work, building on the ideas of Nagarajan et al.~and the general round-or-cut framework of Chakrabarty and Negahbani (2019) that gives tight $3$-approximation algorithms for many Robust $k$-Center variants on general metrics, we present a $(1 + \sqrt{3})$-approximation algorithm for the Robust Euclidean $k$-Supplier problem, thus improving upon the aforementioned $3$-approximation algorithm for Robust $k$-Supplier on general metrics and matching the best approximation factor known for the non-robust setting.

\end{abstract}

%% file: introduction.tex
\section{Introduction}

The $k$-Center problem is one of the most popular and heavily studied clustering problems. In its most basic form, it is defined as follows. Given a finite metric space $(X,d)$ and an integer parameter $k \in \mathbb{N}_{>0}$, the goal is to compute a set of centers $S \subseteq X$ with $|S| \leq k$ so as to minimize $\max_{u \in X} \min_{s \in S} d(s,u)$. Equivalently, the goal is to find the minimum radius $r \geq 0$ and a set of centers $S \subseteq X$ with $|S| \leq k$ such that $B_X(S,r) = X$, where $B_X(S,r) \coloneqq \bigcup_{s \in S} B_X(s,r)$ and $B_X(s,r) \coloneqq \{u \in X: d(s,u) \leq r\}$.

The $k$-Center problem is now well understood, and there exist various $2$-approximation algorithms~\cite{DBLP:journals/tcs/Gonzalez85, DBLP:journals/jacm/HochbaumS86} for it. The problem is also known to be $(2-\varepsilon)$-hard to approximate~\cite{DBLP:journals/dam/HsuN79}, and thus, its complexity is settled, assuming that $\mathtt{P} \neq \mathtt{NP}$. Multiple variants of the problem have been studied during the last four decades. 

In this work, we study two variants of the $k$-Center problem that are formulated in Euclidean settings, namely the Matroid Center problem on the real line and the Robust Euclidean $k$-Supplier problem. We start with the definition of the Matroid Center problem on general metrics.

\begin{definition}[Matroid Center]
Let $(X,d)$ be a finite metric space and $\mathcal{M} = (X, \mathcal{I})$ be a matroid defined on $X$. The goal is to find the minimum radius $r \geq 0$ and an independent set $S \in \mathcal{I}$ such that $B_X(S,r) = X$.
\end{definition}

We clarify that from now on $X$ is allowed to be a multiset; in other words, $d$ can be a pseudometric, and in particular, for Euclidean settings we allow multiple distinct points at the same location. Thus, from now on, whenever we say that $(X,d)$ is a finite metric space, we allow for multiple distinct points being at distance $0$ from each other.

Chen et al.~\cite{DBLP:journals/algorithmica/ChenLLW16} studied the Matroid Center problem and gave an elegant $3$-approximation algorithm for it, which is also tight, as the $k$-Center problem with forbidden centers is a special case of Matroid Center, for which Chuzhoy et al.~\cite{DBLP:journals/jacm/ChuzhoyGHKKKN05} proved that it is $(3-\varepsilon)$-hard to approximate. 

A natural generalization of $k$-Center is the $k$-Supplier problem, introduced by Hochbaum and Shmoys~\cite{DBLP:journals/jacm/HochbaumS86}. In the $k$-Supplier problem, we are given a set of clients that need to be covered, and a set of facilities that can be opened, and the goal is to open at most $k$ facilities such that all clients are covered with the minimum possible radius. The $k$-Supplier problem is also well understood, and straightforward modifications of the aforementioned $2$-approximation algorithms for $k$-Center give a $3$-approximation for $k$-Supplier, which is again tight, assuming that $\mathtt{P} \neq \mathtt{NP}$~\cite{DBLP:journals/jacm/HochbaumS86}. However, there are many applications in which there is either no need to cluster all the points, or the input set of points is noisy, and thus one would need to discard some points before clustering the remaining ones. Many works have studied such variants, which are known as Robust $k$-Center/Supplier, or $k$-Center/Supplier with outliers. In the robust setting, we are given an additional integer parameter $p \in \mathbb{N}_{> 0}$, and the goal is to cluster at least $p$ points. In other words, we are allowed to discard some of the points, which are now treated as outliers. Formally, the Robust $k$-Supplier problem is defined as follows.

\begin{definition}[Robust $k$-Supplier]
Let $X$ be a finite set of points, called clients, and $F$ be a finite set of points, called facilities, such that $(X \cup F, d)$ is a finite metric space. Let $k \in \mathbb{N}_{> 0}$ and $p \in \mathbb{N}_{> 0}$. We denote such an instance as $(X \cup F, k, p)$. The goal is to find the minimum radius $r \geq 0$ and a set of facilities $S \subseteq F$ with $|S| \leq k$ such that $|B_X(S,r)| \geq p$.
\end{definition}

Charikar et al.~\cite{DBLP:conf/soda/CharikarKMN01} gave a very elegant greedy $3$-approximation algorithm for Robust $k$-Supplier, which is optimal, since the problem is a generalization of the $k$-Supplier problem, for which a $3-\varepsilon$ hardness was already known, as mentioned above~\cite{DBLP:journals/jacm/HochbaumS86}.

A natural setting for all these problems that arises in practice is the Euclidean variant, where the input points are points of the Euclidean space $\mathbb{R}^d$ and the distance is the standard $\ell_2$ Euclidean norm. Regarding Matroid Center, Chen et al.~proved that even when restricted to real-line metrics (i.e., when $d = 1$), Matroid Center remains $(2-\varepsilon)$-hard to approximate. This is an interesting result, as it is one of the rare instances of a $k$-Center problem that remains $\mathtt{NP}$-hard even when restricted to the real line. Moreover, the Euclidean $k$-Center and $k$-Supplier problems are not as well understood, as the hardness results do not apply anymore. Feder and Greene~\cite{DBLP:conf/stoc/FederG88} showed that Euclidean $k$-Center is $(\sqrt{3}- \varepsilon)$-hard to approximate, and Euclidean $k$-Supplier is $(\sqrt{7}-\varepsilon)$-hard to approximate, assuming that $\mathtt{P} \neq \mathtt{NP}$. Both hardness results apply even to two-dimensional Euclidean metrics. However, even for two-dimensional Euclidean metrics, the best-known approximation factor for $k$-Center remains $2$. Nagarajan et al.~\cite{nagarajan2020euclidean} gave the first, and only, so far, improvement for Euclidean $k$-Supplier. In particular, they gave a very elegant $(1 + \sqrt{3})$-approximation for Euclidean $k$-Supplier by exploiting a simple property of the Euclidean space and reducing the problem to a minimum Edge Cover computation. Unfortunately, their algorithm is not known to extend to the more general robust setting.

\paragraph{Our results.} In this work, we address Matroid Center on the real line and Robust Euclidean $k$-Supplier and obtain improved approximation algorithms for both problems. Formally, we prove the following theorems.

\begin{theorem}\label{thm:matroid-center}
There exists a $2.5$-approximation algorithm for Matroid Center on the real line.
\end{theorem}

\begin{theorem}\label{thm:k-supplier}
There exists a $(1 + \sqrt{3})$-approximation algorithm for Robust Euclidean $k$-Supplier.
\end{theorem}

To the best of our knowledge, the $3$-approximation algorithm of Chen et al.~\cite{DBLP:journals/algorithmica/ChenLLW16} for Matroid Center on general metrics is the only algorithm known even for Matroid Center on the real line. Thus, Theorem~\ref{thm:matroid-center} improves upon this factor of $3$ for the $1$-dimensional setting and narrows the gap between the best approximation factor and the hardness of $2-\varepsilon$. Moreover, Theorem~\ref{thm:k-supplier} gives a $(1 + \sqrt{3})$-approximation algorithm for Robust Euclidean $k$-Supplier, thus matching the approximation guarantee of the algorithm of Nagarajan et al.~\cite{nagarajan2020euclidean} that only applies to the non-robust setting, and nearly matching the known hardness of $\sqrt{7}-\varepsilon$~\cite{DBLP:conf/stoc/FederG88}. We note that prior to our work, there was no known algorithm for Robust Euclidean $k$-Supplier other than the $3$-approximation algorithm for Robust $k$-Supplier that works on general metrics~\cite{DBLP:conf/soda/CharikarKMN01}.

\paragraph{Note.} An independent recent work of Lee, Nagarajan and Wang~\cite{new_paper} obtained a $(1+\sqrt{3})$-approximation algorithm for Robust Euclidean $k$-Supplier. Our approach is similar to theirs, and at a high level essentially identical, but there are technical differences in both the polytope used for the round-or-cut algorithm as well as in the algorithm that solves the constraint version of Edge Cover that appears as an intermediate problem during the iterations of the ellipsoid algorithm.

\paragraph{Overview of techniques.} At a high level, we use the classic threshold technique introduced by Hockbaum and Shmoys~\cite{DBLP:journals/jacm/HochbaumS86} that reduces the original task of designing an $\alpha$-approximation algorithm to the following task. Given a candidate radius $r$, the goal is to design a procedure that either returns a feasible solution of radius $\alpha r$, or certifies that there is no solution of radius $r$. Based on this, we proceed to design such procedures for Matroid Center on the real line and for Robust Euclidean $k$-Supplier.

Regarding Matroid Center, we closely follow the algorithm of Chen et al.~\cite{DBLP:journals/algorithmica/ChenLLW16}, and we exploit the very special structure of real-line metrics in order to improve upon one of the steps of the algorithm; this allows us to push the approximation factor down to $2.5$. More specifically, following the ideas of \cite{DBLP:journals/jacm/HochbaumS86}, Chen et al.~first greedily partition the space using balls of radius $2r$. This ensures that the centers of this partition are at pairwise distance strictly larger than $2r$. In case $r$ is a feasible radius, the centers of this partition are served by distinct centers in any feasible solution of radius $r$, and thus, Chen et al.~observe that they can now reduce the problem to a matroid intersection problem, which is polynomially solvable. Our contribution lies in observing that for any real-line metric, we can partition the space using balls of radius $1.5r$ such that the centers of this  partition are at pairwise distance strictly larger than $2r$.

As for Robust Euclidean $k$-Supplier, we closely follow the work of Chakrabarty and Negahbani~\cite{chakrabarty2018generalized}, who presented a very general round-or-cut framework that gives tight $3$-approximation algorithms for many robust variants of $k$-Center/Supplier. At a high level, Chakrabarty and Negahbani follow the framework of \cite{DBLP:journals/jacm/HochbaumS86}, and for a given candidate radius $r \geq 0$, they apply the ellipsoid method to the combinatorial polytope of radius $r$ (i.e., the convex hull of integer solutions of radius at most $r$), generate a preliminary clustering of the clients of radius $2r$ based on the candidate fractional point considered at the current iteration of the ellipsoid method, and reduce the problem of separation to a certain decision problem. Then, they show that if one can solve the resulting decision problem, then one can either obtain a feasible solution of radius $3r$ or find a separating hyperplane. It turns out that for many of the robust $k$-Center variants considered, this decision problem admits a polynomial-time algorithm. For that, a crucial property that Chakrabarty and Negahbani exploit is that, as in the Matroid Center case and the Hockbaum and Shmoys~\cite{DBLP:journals/jacm/HochbaumS86} framework, the centers of the candidate clusters obtained from the preliminary clustering are at pairwise distance larger than $2r$. We follow a similar approach, but we incorporate the ideas of Nagarajan et al.~\cite{nagarajan2020euclidean} in order to improve the approximation factor. More specifically, we perform the preliminary clustering with a smaller radius, namely $\sqrt{3} \cdot r$, as in~\cite{nagarajan2020euclidean}. The resulting decision problem to be solved now is more involved, as we do not have the property that the preliminary cluster centers that we have selected are served by distinct facilities of any feasible solution of radius $r$. To overcome this, we exploit the observations of~\cite{nagarajan2020euclidean} about Euclidean metrics in order to reduce the problem to a constraint Edge Cover problem (compared to just a standard minimum Edge Cover problem, as in~\cite{nagarajan2020euclidean}), which we call Max $k$-Edge Cover problem. Nevertheless, we observe that Max $k$-Edge Cover still admits a polynomial-time algorithm by reducing it to a Maximum Weight Perfect Matching computation. Putting everything together, we are able to solve the Max $k$-Edge Cover problem, and thus manage to either round or cut during each ellipsoid iteration.

\paragraph{Organization of material} The rest of this work is organized as follows. In Section~\ref{sec:hockbaum-shmoys}, we give a brief overview of the framework of Hockbaum and Shmoys~\cite{DBLP:journals/jacm/HochbaumS86}, as both of our algorithms follow it; readers familiar with it can skip this section and move directly to Section~\ref{sec:matroid-center}. In Section~\ref{sec:matroid-center}, we present the $2.5$-approximation algorithm for Matroid Center on the real line, while in Section~\ref{sec:k_supplier} we present the $(1+\sqrt{3})$-approximation for Robust Euclidean $k$-Supplier.

%% file: hockbaum_shmoys_framework.tex
\section{The Hockbaum-Shmoys threshold framework}\label{sec:hockbaum-shmoys}

In this section, we briefly review the framework of Hockbaum and Shmoys~\cite{DBLP:journals/jacm/HochbaumS86} for designing approximation algorithms for $k$-Center problems, since both of our algorithms build on the ideas of~\cite{DBLP:journals/jacm/HochbaumS86}. The key observation of~\cite{DBLP:journals/jacm/HochbaumS86} is that in a $k$-Center instance defined on a finite metric space $(X,d)$, the optimal radius $\mathtt{OPT}$ is necessarily equal to the distance between two points of $X$. More formally, we must have $\mathtt{OPT} \in \mathcal{D} \coloneqq \{d(u,v): u,v \in X\}$. It is easy to see that $|\mathcal{D}| = O(|X|^2)$, which implies that any algorithm can start by ``guessing" the optimal radius $\mathtt{OPT}$. Thus, given the set $\mathcal{D}$, the framework proposed by~\cite{DBLP:journals/jacm/HochbaumS86} for obtaining an $\alpha$-approximation algorithm, for some $\alpha \geq 1$, is the following. Let $r \in \mathcal{D}$. Then, it is sufficient to design a polynomial-time procedure that does one of the following:
\begin{enumerate}
    \item it either returns a feasible solution of radius $\alpha r$, or
    \item it certifies that there is no feasible solution of radius $r$.
\end{enumerate}
Given such a procedure, we can go over all candidate radii in $\mathcal{D}$, and among all obtained solutions, return the solution with the smallest radius. By the guarantee of the above procedure, we know that the procedure will return a feasible solution when given $r = \mathtt{OPT}$, and so the final solution obtained will be of radius at most $\alpha \cdot \mathtt{OPT}$. Thus, we obtain an $\alpha$-approximation for the problem.

Using the above framework, Hockbaum and Shmoys obtained a simple $2$-approximation for $k$-Center and a $3$-approximation for $k$-Supplier, which we now briefly describe. For $k$-Center, given a candidate radius $r$, they partition the metric space by greedily picking points as centers and removing balls of radius $2r$ around them. Let $s_1, \ldots, s_{k'}$ be the centers obtained this way. By construction, the distance between any two centers among $s_1, \ldots, s_{k'}$ is strictly larger than $2r$. Thus, if $r$ is a feasible radius, then in any solution of radius $r$, each $s_i$ among the selected ones must be ``served" by a distinct center. This immediately implies that $k' \leq k$ and, moreover, the set of centers $s_1, \ldots, s_{k'}$ is already a feasible solution of radius $2r$.

Transferring this argument to the supplier setting requires one minor modification. The preliminary clustering is performed on the client side, and thus, the set of ``centers" $s_1, \ldots, s_{k'}$ are clients. In order to obtain a feasible solution, the final step of the algorithm selects a facility within distance $r$ from each point $s_1, \ldots, s_{k'}$. If $r$ is a feasible radius, then such facilities are guaranteed to exist and must be distinct, which again implies that $k' \leq k$. Due to this shifting of centers, the final approximation is $3$.

%% file: matroid_center.tex
\section{A \texorpdfstring{$2.5$}{2.5}-approximation algorithm for Matroid Center on the real line}\label{sec:matroid-center}

In this section we prove Theorem~\ref{thm:matroid-center}, which is an improvement upon the current best known 3-ap\-prox\-i\-ma\-tion algorithm of Chen et al.~\cite{DBLP:journals/algorithmica/ChenLLW16} that applies to general metrics. Throughout this section, we use the following notation: $(X,d)$ denotes a finite metric space and $\mathcal{M} = (X, \mathcal{I})$ is a matroid defined on $X$.

At a high level, our algorithm closely follows the algorithm of Chen et al.~\cite{DBLP:journals/algorithmica/ChenLLW16}, who build on the framework of Hockbaum and Shmoys that was presented in Section~\ref{sec:hockbaum-shmoys}. More precisely, for any candidate radius $r\geq 0$, the algorithm in~\cite{DBLP:journals/algorithmica/ChenLLW16} selects a set $S \subseteq X$ that satisfies the following two properties:
\begin{enumerate}
    \item $d(s, s') > 2r$ for every $s \neq s' \in S$,
    \item $d(S, x) \coloneqq \min_{s \in S} d(s,x) \leq \alpha r$ for every $x \in X$,
\end{enumerate}
where $\alpha \geq 1$ is some parameter to be specified. If $r$ is a feasible radius, then property (1) implies that any two points in $S$ are served by distinct centers in any feasible solution of radius $r$. Thus, \cite{DBLP:journals/algorithmica/ChenLLW16} proceed to define another matroid $\mathcal{M}' = (X, \mathcal{I}')$, where $I \in \mathcal{I}'$ if and only if $|I \cap B_X(s,r)| \leq 1$ for every $s \in S$ and $|I \cap (X \setminus B_X(S,r))| \leq 0$, and so the resulting matroid $\mathcal{M}'$ is a partition matroid on $X$. It is easy to see now that if we compute a maximum cardinality independent set that belongs to both matroids $\mathcal{M}$ and $\mathcal{M}'$, we will obtain a feasible solution of radius $(1+\alpha)r$. Since there are polynomial-time algorithms for the matroid intersection problem, this is exactly what \cite{DBLP:journals/algorithmica/ChenLLW16} do and they obtain a $(1+\alpha)$-approximation algorithm for the problem. Finally, in order to specify $\alpha$, they mimic the algorithm of Hockbaum and Shmoys and use a greedy partitioning of the whole space with balls of radius $2r$. Thus, they use $\alpha = 2$ and obtain a $3$-approximation algorithm for general metrics, which, as noted in the introduction, is tight.

Our main contribution is that we show that on the real line, we can reduce the constant $\alpha$ from $2$ down to to $1.5$, which in turn improves the approximation factor from $3$ to $2.5$. More formally, we prove the following structural lemma, where we define $d(x,x') \coloneqq |x - x'|$ for every $x,x' \in \mathbb{R}$.

\begin{lemma}\label{lemma:real-line}
Let $r \geq 0$, and let $X = \{x_1, \ldots, x_n\} \subseteq \mathbb{R}$, such that $x_i \leq x_{i+1} \leq x_i + r$, for every $i \in [n-1]$. Then, there exists a set $S \subseteq X$ such that:
\begin{enumerate}
    \item $d(s, s') > 2r$ for every $s \neq s' \in S$,
    \item $d(S, x) \leq 1.5 r$ for every $x \in X$.
\end{enumerate}
Moreover, such a set $S$ can be computed in $\mathtt{poly}(|X|)$ time.
\end{lemma}

Before proving the above lemma, we observe that if we require property (1), then the factor $\alpha = 1.5$ in property (2) is tight. An example demonstrating this is the following. For any given $r > 0$, let $X_r = \{0, 0.5r, 1.5r, 2r\}$. Note that the set $S = \{0.5r\}$ (or, alternatively, $S = \{1.5r\}$) satisfies property (1) of Lemma~\ref{lemma:real-line}, and also satisfies property (2), since we have $\max_{x' \in X_r} |0.5r - x'| \leq 1.5r$. Moreover, any set $S \subseteq X$ with $|S| \geq 2$ clearly violates property (1). This shows that the factor $1.5$ in property (2) of Lemma~\ref{lemma:real-line} is tight. Thus, in order to further improve the approximation factor, new ideas are needed.

Assuming Lemma~\ref{lemma:real-line}, whose proof is given in Section~\ref{sec:proof-real-line}, we can now prove the following theorem, which is the key theorem that will allow us to design our algorithm.

\begin{theorem}\label{thm:real-line-either-or}
Let $X$ be a finite subset of $\mathbb{R}$, and let $M = (X, \mathcal{I})$ be a matroid defined on $X$. There is an algorithm running in $\mathtt{poly}(|X|)$ time that, given a candidate radius $r \in \mathbb{R}_{\geq 0}$, does one of the following:
\begin{enumerate}
    \item it either returns a solution of radius $2.5r$, or
    \item it certifies that there is no solution of radius $r$.
\end{enumerate}
\end{theorem}

\begin{proof}
Let $X = \{x_1, \ldots, x_n\} \subseteq \mathbb{R}$ such that $x_1 \leq x_2 \leq \ldots \leq x_n$. We first process the points from left to right as follows. Let $i \in [n-1]$ be the smallest index, if any, such that $x_{j+1} - x_j \leq r$, for every $j \in [i-1]$, and $x_{i + 1} - x_i > r$; if there is no such index, then we set $i = n$. Note that we have $i = 1$ if $x_2 - x_1 > r$. Let $X_1 = \{x_1, \ldots, x_i\}$. We remove the set $X_1$ from $X$, and repeat the same procedure, starting from the index $i+1$, until we exhaust all points of $X$. Thus, we end up with a partition $X_1, \ldots, X_t$ of $X$, where $t \in \mathbb{N}_{\geq 1}$. 

We now apply Lemma~\ref{lemma:real-line} to each set $X_i$, $i \in [t]$, since each such set satisfies the condition of the lemma with respect to the given radius $r$, and obtain sets $S_1, \ldots, S_t$, where $S_i \subseteq X_i$ for every $i \in [t]$. We then define a partition matroid $M' = (X, \mathcal{I}')$ by considering the following partition of $X$:
\begin{itemize}
    \item For each $i \in [t]$ and $s \in S_i$, each set $B_X(s_i,r)$ is a set of the partition.
    \item The set $X \setminus \bigcup_{i \in [t]} B_X(S_i,r)$ belongs to the partition.
\end{itemize}
The above sets indeed define a partition of $X$. To see this, note that for any $i \in [t]$ and $s \neq s' \in S_i$, we have $B_X(s,r) \cap B_X(s',r) = \emptyset$, by property (1) of Lemma~\ref{lemma:real-line}. Moreover, if $s \in S_i$ and $s' \in S_j$, for $1 \leq i < j \leq t$, then if $x \in X_i$ is the rightmost point in $X_i$ and $x' \in X_j$ is the leftmost point in $X_j$, we have $|x - x'| > r$. This shows that $B_X(s,r) \cap B_X(s',r) = \emptyset$. We conclude that the sets described above form a partition of $X$. We are now ready to define the independent sets in $\mathcal{I}'$. We have $I \in \mathcal{I}'$ if and only if:
\begin{itemize}
    \item $|I \cap B_X(s,r)| \leq 1$, for every $i \in [t]$ and $s \in S_i$, and
    \item $|I \cap (X \setminus \bigcup_{i \in [t]} B_X(S_i,r))| \leq 0$.
\end{itemize}

The last step of our algorithm, similar to \cite{DBLP:journals/algorithmica/ChenLLW16}, is to compute an independent set $\bar{I} \in \mathcal{M} \cap \mathcal{M}'$ of largest possible cardinality. This can be done in polynomial time, since the matroid intersection problem admits polynomial-time algorithms. Thus, we end up with a set $\bar{I}$. If $|\bar{I} \cap B_X(s,r)| = 1$ for every $i \in [t]$ and $s \in S_i$, then we return the set $\bar{I}$, which we claim is a $2.5$-approximation, while if $|\bar{I} \cap B_X(s,r)| = 0$ for some $i \in [t]$ and $s \in S_i$, the algorithm outputs that there is no solution of radius $r$.

Suppose that $|\bar{I} \cap B_X(s,r)| = 1$ for every $i \in [t]$ and $s \in S_i$. Observe that for any $i \in [t]$ and any point $x \in X_i$, we have $d(s, x) \leq 1.5r$ for some $s \in S_i$; this holds by property (2) of Lemma~\ref{lemma:real-line}. Thus, since $|\bar{I} \cap B_X(s,r)| = 1$, this means that there exists $\bar{s} \in \bar{I}$ such that $d(s, \bar{s}) \leq r$. By the triangle inequality, this now implies that $d(\bar{s},x) \leq 1.5r + r = 2.5r$. Thus, we conclude that in this case, we have $B_X(\bar{I}, 2.5r) = X$.

So, the only thing remaining to prove is that if $r$ is a feasible radius, then our algorithm will always succeed, i.e., we will always have $|\bar{I} \cap B_X(s,r)| = 1$ for every $i \in [t]$ and $s \in S_i$. To prove this, let $F \subseteq X$ be a feasible solution of radius $r$, which implies that $F \in \mathcal{I}$. We have $d(F,s) \leq r$ for every $i \in [t]$ and $s \in S_i$. In particular, this implies that $|F \cap B_X(s,r)| \geq 1$ for every $i \in [t]$ and $s \in S_i$. Since $B_X(s,r) \cap B_X(s',r) = \emptyset$ for any $s \neq s' \in \bigcup_{i \in [t]} S_i$, there exists $F' \subseteq F$ such that $|F' \cap B_X(s,r)| = 1$ for every $i \in [t]$ and $s \in S_i$; clearly, $F' \in \mathcal{I}$. Given the constraints defining the partition matroid $M'$, it is easy to see that $\bar{F} = F' \cap \left(\bigcup_{i \in [t]} B_X(S_i,r) \right)$ is an independent set in $\mathcal{M} \cap \mathcal{M}'$ of maximum cardinality. Thus, the matroid intersection algorithm used must necessarily return a set of the same cardinality, which implies that $\bar{I}$ will satisfy  $|\bar{I} \cap B_X(s,r)| = 1$ for every $i \in [t]$ and $s \in S_i$. We conclude that if $r$ is a feasible radius, then the algorithm will always return a solution of radius $2.5r$. This finishes the proof.
\end{proof}

We are now ready to prove Theorem~\ref{thm:matroid-center}.
\begin{proof}[Proof of Theorem~\ref{thm:matroid-center}]
Let $X$ be a finite subset of $\mathbb{R}$, and let $\mathcal{M} =  (X, \mathcal{I})$ be a matroid defined on $X$. Let $\mathcal{D} = \{|x-x'|: \; x,x' \in X\}$. Observe that $|\mathcal{D}| = O(|X|^2)$ and, moreover, $\mathcal{D}$ can be computed in $\mathtt{poly}(|X|)$ time. Let $\mathtt{OPT}$ be the optimal radius of the given instance; we have $\mathtt{OPT} \in \mathcal{D}$.

We now go over all distances in $\mathcal{D}$, and for each $r \in \mathcal{D}$, we use the polynomial-time procedure of Theorem~\ref{thm:real-line-either-or}, and among all computed solutions, we return the solution with the smallest radius. Since $\mathtt{OPT} \in \mathcal{D}$, it is clear that we will compute at least one solution, and the returned solution will have radius at most $2.5 \cdot \mathtt{OPT}$, because of the guarantees of Theorem~\ref{thm:real-line-either-or}. Moreover, since $|\mathcal{D}| = O(|X|^2)$, the whole procedure will run in $\mathtt{poly}(|X|)$ time. This concludes the proof.
\end{proof}

\subsection{The proof of Lemma~\ref{lemma:real-line}}\label{sec:proof-real-line}

If $r = 0$, then the statement holds by setting $S = X$. So, from now on we assume that $r > 0$.  We begin by scaling each point in $X$ by $r^{-1}$; more precisely, we set $x_i' = \frac{x_i}{r}$ for every $i \in [n]$. We will prove that the statement holds for the set $X' = \{x_1', \ldots, x_n'\}$, where we have $x_i' \leq x_{i+1}' \leq x_i' + 1$ for every $i \in [n-1]$. In particular, we will show that there is a set $S' \subseteq X'$ that can be computed in $\mathtt{poly}(n)$ time such that:
\begin{enumerate}
    \item $d(s,s') > 2$ for all $s\neq s' \in S'$,
    \item $d(S',x') \leq 1.5$, for every $x' \in X'$.
\end{enumerate}
If this is the case, then it is easy to see that the set $S = \{rs': \; s' \in S'\} \subseteq X$ satisfies the desired properties with respect to the original input set $X$. Thus, without loss of generality, from now on we assume that $x_i \leq x_{i+1} \leq x_i + 1$ for every $i \in [n]$.

We now define a directed acyclic graph whose paths describe ways of selecting candidate sets with the desired properties. Let $G = (V, E), $ where $V = X \cup \{t, \overline t\}$ and the set of edges $E$ is defined as follows.
\begin{itemize}
\item There is an edge from $t$ to $x_i$, $i \in [n]$, if $x_i - x_1 \leq 1.5$.
\item There is an edge from $x_i$, $i \in [n]$, to $\overline t$, if $x_n - x_i \leq 1.5$.
\item There is an edge from $x_i$ to $x_j$, $i < j$, if $x_j - x_i > 2$ and $d(\{x_i, x_j\}, x_l) \leq 1.5$, for every $i < l < j$.
\end{itemize}

If there exists a directed path from $t$ to $\overline t$ in the graph $G$, then the points along such a path constitute a set $S$ that satisfies the two properties of Lemma~\ref{lemma:real-line}. We will prove the existence of such a path by repeated application of the following lemma.

\begin{lemma} \label{mat:forward}
Let $x\in X$ such that all points in $X \cap [x, x+1.5]$ are reachable from $t$ in $G$. Let $x' \in X$ be the smallest point that is larger than $x+2$, if any. Then, all points in $X \cap [x', x'+1.5]$ are reachable from $t$.
\end{lemma}

\begin{proof}
If there is no point in $X$ larger than $x + 2$, then there is nothing to prove. So, we assume that there is at least one such point, and $x'$ is the smallest such point. We first show that $x'$ is reachable from $t$. For that, we do some case analysis. First, assume that $|x - x'| > 3$. In such a case, there exists a point $\bar x \in V$ with $\bar x < x'$ such that $|\bar x - x'| \leq 1$ and consequently $|x - \bar x| > 2$. This contradicts that $x'$ is the smallest point in $X$ larger than $x+2$. Therefore, $|x - x'| \leq 3$, which implies that any point $w \in X \cap [x, x']$ is of distance at most $1.5$ to either $x$ or $x'$. Hence there is an edge from $x$ to $x'$.

We now show that any $w' \in X \cap (x', x'+1.5]$ is reachable from $t$. Suppose that this is not the case, and let $w' \in X \cap (x', x'+1.5]$ be a point that is not reachable from $t$. We observe that we must have $|x - w'| > 3$, since otherwise the argument in the previous paragraph shows that there will be an edge from $x$ to $w'$. Let $w$ be the largest point in $X \cap[x, x+1.5]$ such that $|w - w'| > 2$; such a point always exists since $|x-w'| > 3$. As we have assumed that $w'$ is not reachable from $t$, and since $w$ is reachable from $t$ by assumption, there is no edge from $w$ to $w'$, and so there exists a point $u \in X$ between $w$ and $w'$ such that $d(\{w, w'\}, u) > 1.5$. This implies that $|w - w'| > 3$, since $|w - w'| = |w - u| + |u - w'|$. It is now easy to see that we must have $w \in X\cap (x+0.5, x+1.5]$; if not, then there would be another point $\bar{w} \in X \cap (w, x+ 1.5]$ satisfying $|\bar{w} - w'| > 2$, since consecutive points are at distance at most $1$ and $|w - w'| > 3$. Since $x < w < u$, we have $|x - u| = |x - w| + |w - u| > 0.5 + 1.5 =  2$. This implies that $x' \leq u$, and thus we have $w < x' \leq u < w'$. Using the fact that $w' \in X \cap (x', x'+1.5]$, we now get that $|u - w'| \leq 1.5$, which gives a contradiction. Putting everything together, we conclude that all points in $X \cap [x', x'+1.5]$ are reachable from $t$.
\end{proof}

We now prove the existence of a directed path from $t$ to $\overline t$ in $G$ by considering all potential distances $x_n - x_1$.
\begin{enumerate}
    \item $x_n - x_1 \leq 1.5$. Such a path exists by construction. 
    \item $1.5 < x_n - x_1 \leq 2$. Let $v \in X$ be the smallest point with $x_n -v \leq 1$. Then $v \neq x_n$ since $x_{i+1}-x_i \leq 1$ for all $i \in [n-1]$. If $v \leq x_1+1.5,$ then there is an edge from $t$ to $v$ and from $v$ to $\overline t$ in $G$, and thus a path from $t$ to $\overline t$. So, suppose that $v > x_1 + 1.5$. In this case, we must have $X \cap [x_1 + 1, x_1 + 1.5] = \emptyset$. Let $v' \in X$ be the largest point with $x_1 - v' \leq 1$. It follows that $v-v' \leq 1$, since there are no other points between $v$ and $v'$. Since $x_n - x_1 \leq 2$ and $v - x_1 > 1.5$, we have $|v - x_n| \leq 0.5$, and so $|v' - x_n| = |v' - v| + |v - x_n| \leq 1 + 0.5 = 1.5$. We conclude that there is an edge from $t$ to $v'$ and from $v'$ to $\overline t$, and thus a path from $t$ to $\overline t$.


  \item $x_n - x_1 > 2$. Let $v \in X$ be the smallest point in $X \cap (x_1+2, x_n]$. By construction, all points in $X \cap [x_1, x_1 + 1.5]$ are reachable from $t$. Thus, we can apply Lemma \ref{mat:forward} and get that all points in $X \cap [v, v+1.5]$ are reachable from $t$. By repeatedly applying Lemma \ref{mat:forward} to the smallest point $v' \in X$ with $v' \in (v+2, x_n]$, we will eventually find a point $\overline v$ such that all points in $X \cap [\overline v, \overline v + 1.5]$ are reachable from $t$ and $x_n - \overline v \leq 2$. The previous case analysis will conclude the proof.
\end{enumerate}

%% file: k_supplier.tex
\section{A \texorpdfstring{$(1 + \sqrt{3})$}{1+sqrt3}-approximation algorithm for Robust Euclidean \texorpdfstring{$k$}{k}-Supplier}\label{sec:k_supplier}

In this section, we prove Theorem~\ref{thm:k-supplier}. Throughout this section, we use the following notation. Let $X \subset \mathbb{R}^d$ denote a finite set of clients, let $F\subset \mathbb{R}^d$ denote a finite set of facilities, and let $d(u,v) = \|u - v\|_2$ for every $u,v \in X \cup F$. Let $k \in \mathbb{N}_{>0}$ be the maximum number of centers allowed to be opened, and let $p \in \mathbb{N}_{\geq 0}$ be the number of clients that need to be covered.

Similar to Matroid Center, our algorithm again builds on the framework of Hockbaum and Shmoys that was presented in Section~\ref{sec:hockbaum-shmoys}. In particular, the main building block is the following theorem, which we prove in Section~\ref{sec:proof-aux}.

\begin{theorem}\label{thm:k-supplier-either-or}
Let $(X, F, k, p)$ be an instance of Robust Euclidean $k$-Supplier. There is a polynomial-time algorithm that, given a candidate radius $r \in \mathbb{R}_{\geq 0}$, does one of the following:
\begin{enumerate}
\item it either returns a feasible solution of radius $(1+\sqrt 3)r$, or
\item it certifies that there is no solution of radius $r$.
\end{enumerate}
\end{theorem}

Equipped with the above theorem, we are now ready to prove Theorem~\ref{thm:k-supplier}.

\begin{proof}[Proof of Theorem~\ref{thm:k-supplier}]
Let $(X,F,k,p)$ be an instance of Robust Euclidean $k$-Supplier, where $X \cup F \subset \mathbb{R}^d$, for some $d \in \mathbb{N}_{\geq 1}$ that is allowed to be part of the input. Let $\mathcal{D} = \{\|x-f\|_2: \; x \in X,f \in F\}$. Observe that $|\mathcal{D}| = O(|X| \cdot |F|)$ and, moreover, $\mathcal{D}$ can be computed in polynomial time. Let $\mathtt{OPT}$ be the optimal radius of the given instance; we have $\mathtt{OPT} \in \mathcal{D}$.

We now go over all distances in $\mathcal{D}$, and for each $r \in \mathcal{D}$, we use the polynomial-time procedure of Theorem~\ref{thm:k-supplier-either-or}, and among all computed solutions, we return the solution with the smallest radius. Since $\mathtt{OPT} \in \mathcal{D}$, it is clear that we will compute at least one solution, and the returned solution will have radius at most $(1+\sqrt{3}) \mathtt{OPT}$, because of the guarantees of Theorem~\ref{thm:k-supplier-either-or}. Moreover, since $|\mathcal{D}| = O(|X| \cdot |F|)$, the whole procedure will run in polynomial time. This concludes the proof.
\end{proof}


\subsection{Proof of Theorem~\ref{thm:k-supplier-either-or}}\label{sec:proof-aux}

We now discuss the algorithm and proof of Theorem~\ref{thm:k-supplier-either-or}. At a high level, our algorithm employs the round-or-cut framework, first introduced by Carr et al.~\cite{carr1999strengthening}, and later used in Robust $k$-Center variants by Chakrabarty and Negahbani \cite{chakrabarty2018generalized}. For a given radius $r \geq 0$, the round-or-cut method uses the ellipsoid method on the combinatorial polytope (i.e., the convex hull of integer solutions of radius at most $r$) and in each iteration we either round a candidate fractional point (provided by the ellipsoid algorithm) into a solution of radius $(1+\sqrt 3)r$, or separate this point from the combinatorial polytope. In particular, for each such candidate point, we follow the strategy of~\cite{chakrabarty2018generalized} that first clusters all the clients in a certain way; we note that this clustering first appeared in Harris et al.~\cite{harris2017lottery}. Given the resulting clustering, we then construct what we call a Max $k$-Edge Cover instance, inspired by the techniques of Nagarajan et al.~\cite{nagarajan2020euclidean}, and show that this auxiliary problem can be solved in polynomial time. If the candidate point provided by the ellipsoid algorithm is in the combinatorial polytope, we show that the Max $k$-Edge Cover problem admits a solution of weight at least $p$, that directly corresponds to a solution of radius $(1+\sqrt 3)r$ solution. On the other hand, in the case the optimal value of the Max $k$-Edge Cover problem is smaller than $p$, we show how to construct a hyperplane separating the current point from the polytope. The ellipsoid method ensures that in polynomial time, we can either conclude that the combinatorial polytope is empty, and thus $r$ is not a feasible radius, or construct a solution of radius $(1+\sqrt 3)r$. Since the subroutine for solving the Max $k$-Edge Cover problem is polynomial-time, this implies that the whole algorithm also runs in polynomial time.

We now explain the details of the algorithm. Let $(X \cup F, k, p)$ be an instance of Robust Euclidean $k$-Supplier. Let $r \geq 0$ be the given candidate radius. We define the following polytope, which corresponds to the canonical relaxation of our problem, where the $x$-vector indicates whether a client participates in the clustering, and the $y$-vector indicates which facilities are selected as centers.
\begin{equation*}\label{eq:defP}
\mathcal{P}(r) = \left\{ (x,y)\in [0,1]^X \times [0,1]^F \: \middle| \:
\renewcommand\arraystretch{1.2}
\begin{array}{>{\displaystyle}rcl@{\quad}l}
    \sum_{s\in F}y(s)                &\leq &k & \\
    \sum_{s\in B_F(u,r)} y(s)  &\geq& x(u) & \forall u\in X \\ 
    \sum_{u\in X} x(u)             &\geq &p &
\end{array}
\right\}.
\end{equation*}

It is easy to verify that $\mathcal{P}(r)$ is indeed a relaxation of our problem. In particular, if $r$ is a feasible radius, then $\mathcal{P}(r)$ is non-empty. We also define $\mathcal{P}_I(r) \coloneqq \mathrm{conv}\left(\mathcal{P}(r) \cap (\{0,1\}^X \times \{0,1\}^F )\right)$. The polytope $\mathcal{P}_I(r)$ is the combinatorial polytope corresponding to the radius $r$, i.e., the convex hull of integer solutions of radius at most $r$.

\begin{observation}\label{obs:subset}
For every $r \geq 0$, $\mathcal{P}_I(r) \subseteq \mathcal{P}(r)$.
\end{observation}

We apply the ellipsoid method to the polytope $\mathcal{P}_I(r)$ to check whether it is empty or not. Of course, this is an $\mathtt{NP}$-hard task, so we are not always able to separate. Nevertheless, we show that in case where we cannot separate, then we will be able to produce a $(1 + \sqrt{3})r$ solution. More specifically, the ellipsoid method will provide a sequence of candidate points $(x,y) \in [0,1]^{X}\times [0,1]^{F}$, and for each point $(x,y)$ considered, we are able to do the following:
\begin{enumerate}
\item we either find a solution of radius at most $(1 + \sqrt{3})r$ and stop the algorithm, or
\item we find a hyperplane separating $(x,y)$ from $\mathcal{P}_I(r)$.
\end{enumerate}
We will ensure that the  encoding length of the generated hyperplanes in the second case is polynomially bounded, and thus, if we are not able to obtain a solution of radius at most $(1 + \sqrt{3})r$, i.e., if the first case never happens, then in polynomially many steps (see Theorem 6.4.9 in \cite{grotschel2012geometric}) the ellipsoid method will certify that $\mathcal{P}_I(r) = \emptyset$. If $r$ is a feasible radius for the given instance, we have that $\mathcal{P}_I(r)  \neq \emptyset$, and so the above procedure will necessarily recover a solution of radius $(1+\sqrt 3)r$ in a polynomial number of steps.

We now proceed to formalize the above discussion. Let $(x,y) \in \mathbb{R}^X \times \mathbb{R}^F$ be a candidate point considered by the ellipsoid method. We first check if $(x,y) \in \mathcal{P}(r)$. If not, then it means that $(x,y)$ violates one of the inequalities in the inequality description of $\mathcal{P}(r)$. Thus, Observation~\ref{obs:subset} implies that we can use the violated inequality as a separating hyperplane. Clearly, its encoding length is polynomial.

So, from now on we assume that $(x,y) \in \mathcal{P}(r)$. We follow the strategy of Chakrabarty and Negahbani~\cite{chakrabarty2018generalized}. We use their $x$-based clustering, which first appeared in a prior work of Harris et al.~\cite{harris2017lottery}, in order to obtain a set of cluster representatives $C_{\textrm{rep}} \subseteq X$ and corresponding clusters $\{X_c\}_{c \in C_{\text{rep}}}$, where $X_c \subseteq X$ for every $c \in C_{\text{rep}}$. As in~\cite{harris2017lottery,chakrabarty2018generalized}, in order to obtain this clustering, we sort the clients in non-increasing order of their $x$-value and then greedily partition the space $X$ based on these $x$-values. The main difference compared to the clustering of~\cite{harris2017lottery,chakrabarty2018generalized} is that our algorithm considers balls of radius $\sqrt{3} \cdot r$ (similar in spirit to the clustering of Nagarajan et al.~\cite{nagarajan2020euclidean}), as opposed to the balls of radius $2r$ that~\cite{harris2017lottery,chakrabarty2018generalized} use. This results in a different decision problem that we will have to solve based on the clustering obtained, compared to the one that~\cite{chakrabarty2018generalized} deal with. A formal description of the algorithm is given in Algorithm \ref{alg:RFiltering2}.

\begin{algorithm}
  \caption{The $x$-based greedy clustering algorithm.}
  \label{alg:RFiltering2}

Set $Z = \{u \in X: x(u) > 0\}$.\\

Set $C_{\textrm{rep}} = \emptyset$.\\

\While{$Z \neq \emptyset$} {
    Let $c = \arg\max_{u \in Z} x(u)$ (breaking ties arbitrarily).
    
    Set $C_{\textrm{rep}} = C_{\textrm{rep}} \cup \{c\}$.
    
    Let $X_c = B_Z(c, \sqrt{3} \cdot r)$.
    
    Set $Z = Z \setminus X_c$.
}

Return $C_{\textrm{rep}}$ and $\{X_c\}_{c \in C_{\textrm{rep}}}$
\end{algorithm}

It is easy to see that the sets $\{X_c\}_{c \in C_{\textrm{rep}}}$ form a subpartition of $X$, and we also have $d(c, c') > \sqrt{3} \cdot r$ for every $c \neq c' \in C_{\textrm{rep}}$. Moreover, due to the greedy choices of the algorithm, for every $c \in C_{\textrm{rep}}$, we have $x(c) \geq x(v)$ for every $v \in X_c$. Since $(x,y) \in \mathcal{P}(r)$, this implies that
\begin{equation}\label{eq:covering}
    p \leq \sum_{u \in X} x(u) = \sum_{u \in X:\; x(u) > 0} x(u) = \sum_{c \in C_{\textrm{rep}}} \sum_{u \in X_c} x(u) \leq \sum_{c \in C_{\textrm{rep}}} |X_c| \cdot x(c).
\end{equation}

Based on the sets $\{X_c\}_{c \in C_{\textrm{rep}}}$, we now use a reduction very similar to the one in Nagarajan et al.~\cite{nagarajan2020euclidean} in order to reduce our problem to a maximization version of the Edge Cover problem (formally defined later on). For that, we first state a simple lemma of~\cite{nagarajan2020euclidean} that exploits the Euclidean metric, and justifies why we can reduce our problem to an Edge Cover problem.

\begin{lemma}[Lemma 1 in~\cite{nagarajan2020euclidean}]
\label{makeagraph}
For any facility $f \in F$, the number of clients in $C_{\text{rep}}$ that are within distance $r$ from $f$ is at most 2.
\end{lemma}

The proof of the above lemma is very simple, and is based on the fact that  $d(c, c') > \sqrt{3} \cdot r$ for every $c \neq c' \in C_{\textrm{rep}}$, and the following observation. Let $f \in F$ be a facility, and suppose that there are at least three clients $c_1, c_2, c_3 \in C_{\text{rep}}$ within distance $r$ from $f$. We consider the plane that contains $c_1, c_2, c_3 \in C_{\text{rep}}$. Then, there must be a circle on that plane of radius $r$ that contains all three points $c_1, c_2, c_3$. This implies that at least one of the distances $d(c_1, c_2), d(c_1, c_3), d(c_2, c_3)$ is at most $\sqrt{3} \cdot r$, which is a contradiction.

As in~\cite{nagarajan2020euclidean}, we construct the following graph $G = (C_{\textrm{rep}}, E)$. Each edge in the graph represents some facility within distance $r$ to the edge's endpoints. More specifically, for each facility $f \in F$, we consider the set $\{c \in C_{\textrm{rep}}: d(f,c) \leq r\}$, and do the following, depending on the size of this set.
\begin{itemize}
    \item $\{c \in C_{\textrm{rep}}: d(f,c) \leq r\} = \emptyset$: In this case, no edge is added.
    \item $\{c \in C_{\textrm{rep}}: d(f,c) \leq r\} = \{c\}$: In this case, we add a self-loop $e_f = (c,c)$ on the vertex $c$ of $G$.
    \item $\{c \in C_{\textrm{rep}}: d(f,c) \leq r\} = \{c_1, c_2\}$: In this case, we add the edge $e_f = (c_1, c_2)$.
\end{itemize}
Due to Lemma~\ref{makeagraph}, there are no other cases to consider, and thus the resulting $G$  is a graph that might contain parallel edges and self-loops. We note that there is clear injection from the set $E$ of edges to the set $F$ of facilities, and from now on we will use the phrase ``facility $f$ corresponding to an edge $e_f$ of $G$". We also define $w(c) = |X_c|$ for every $c \in C_{\textrm{rep}}$.

Roughly speaking, Nagarajan et al.~\cite{nagarajan2020euclidean} show that if $r$ is a feasible radius for the non-robust version of the problem, a minimum edge cover of $G$ corresponds to a feasible solution of radius $(1+\sqrt 3)r$. However, in the robust setting choosing a feasible set of edges/facilities is not as straightforward. For that, we define the following variant of Edge Cover, which we call Max $k$-Edge Cover. From now on we use the following notation. Given a graph $G = (V, E)$ and a subset of edges $E' \subseteq E$, we write $V(E') \coloneqq \{v \in V: \exists e \in E' \textrm{ s.t. }v \in e\}$; we will say that the set of edges $E'$ covers the set of vertices $V(E')$.

\begin{definition}[Max $k$-Edge Cover]
Let $G = (V, E)$ be an undirected graph, where parallel edges and self-loops are allowed. Let $w: V \to \mathbb{R}_{\geq 0}$, $k \in \mathbb{N}$, and $p \in \mathbb{R}_{\geq 0}$. The goal is to select a set $S \subseteq E$ with $|S| \leq k$ so as to maximize the total weight of the vertices covered by $S$, i.e., the quantity $w(V(S)) \coloneqq \sum_{v \in V(S)} w(v)$.
\end{definition}

In the following two lemmas, we explain why an algorithm for Max $k$-Edge Cover suffices to design an algorithm for our problem.

\begin{lemma}\label{lemma:yes}
Let $G = (C_{\textrm{rep}}, E)$ be the graph, as constructed above using the candidate point $(x,y) \in \mathcal{P}(r)$ and the resulting sets $C_{\textrm{rep}}$ and $\{X_c\}_{c \in C_{\textrm{rep}}}$, and let $w(c) = |X_c|$ for every $c \in C_{\textrm{rep}}$ and $k,p \in \mathbb{N}$. If the resulting Max $k$-Edge Cover instance has a feasible solution of weight at least $p$ that we can compute in polynomial time, then in polynomial time we can construct a feasible solution for the Robust Euclidean $k$-Supplier instance of radius at most $(1 + \sqrt{3})r$.
\end{lemma}
\begin{proof}
Suppose that we can efficiently compute a set $S \subseteq E$ such that
\begin{equation*}
    w(C_{\textrm{rep}}(S)) = \sum_{c \in C_{\textrm{rep}}(S)} w(c) = \sum_{c \in C_{\textrm{rep}}(S)} |X_c| \geq p.
\end{equation*}
Let $F_S \subseteq F$ be the set of facilities corresponding to the edges in $S$. We have $|F_S| = |S| \leq k$. Moreover, by construction, we have $C_{\textrm{rep}}(S) \subseteq B_X(F_S,r)$. By the triangle inequality, this implies that
\begin{equation*}
    \bigcup_{c \in C_{\textrm{rep}}(S)} X_c \subseteq B_X(F_S, (1 + \sqrt{3})r).
\end{equation*}
Since the sets $\{X_c\}_{c \in C_{\textrm{rep}}}$ are pairwise disjoint, we conclude that the set $F_S \subseteq F$ is a feasible set of facilities that gives a solution of radius at most $(1 + \sqrt{3})r$.
\end{proof}

\begin{lemma}\label{lemma:no}
Let $G = (C_{\textrm{rep}}, E)$ be the graph, as constructed above using the candidate point $(x,y) \in \mathcal{P}(r)$ and the resulting sets $C_{\textrm{rep}}$ and $\{X_c\}_{c \in C_{\textrm{rep}}}$, and let $w(c) = |X_c|$ for every $c \in C_{\textrm{rep}}$ and $k,p \in \mathbb{N}$. If the resulting Max $k$-Edge Cover instance has optimal value that is strictly smaller than $p$, then $(x,y) \notin \mathcal{P}_I(r)$ and the inequality $\sum_{c \in C_{\textrm{rep}}} |X_c| \cdot x'(c) \leq p - 1$ is valid separating hyperplane, separating $(x,y)$ from $\mathcal{P}_I(r)$.
\end{lemma}
\begin{proof}
We will show that if every subset $S \subseteq E$ of at most $k$ edges satisfies $w(C_{\textrm{rep}}(S)) < p$, then every point $(x',y') \in \mathcal{P}_I(r)$ satisfies the inequality $ \sum_{c \in C_{\textrm{rep}}} |X_c| \cdot x'(c) \leq p - 1$. Since the candidate point $(x,y)$ is in $\mathcal{P}(r)$, by~(\ref{eq:covering}) it satisfies $\sum_{c \in C_{\textrm{rep}}} |X_c| \cdot x(c) \geq p > p - 1$ , and this will immediately imply the lemma.

So, suppose that there exists a point $(x', y') \in \mathcal{P}_I(r)$ that satisfies $\sum_{c \in C_{\textrm{rep}}} |X_c| \cdot x'(c) > p - 1$. Since $(x', y')$ is a convex combination of integral points, and since the coefficients in the left-hand size of the inequality are all non-negative integer numbers, this means that there must exist an integral point $(\chi^Q,\chi^T) \in \mathcal{P}_I(r)$ that satisfies $\sum_{c \in C_{\textrm{rep}}} |X_c| \cdot \chi^Q(c) \geq p > p - 1$, where $Q \subseteq X$, $T \subseteq F$, $|T| \leq k$, $|Q| \geq p$ and $Q \subseteq B_X(T,r)$; here, $\chi^A$ is the indicator vector of a set $A$. Let $E_T$ be the set of edges corresponding to the facilities in $T$. We have $|E_T| \leq k$. It is easy to see that $C_{\textrm{rep}} \cap Q \subseteq C_{\textrm{rep}}(E_T)$. Thus, we get
\begin{equation*}
    \sum_{c \in C_{\textrm{rep}}(E_T)} w(c)= \sum_{c \in C_{\textrm{rep}}(E_T)} |X_c| \geq \sum_{c \in C_{\textrm{rep}}} |X_c| \cdot \chi^Q(c) \geq p,
\end{equation*}
which is a contradiction, since by assumption we have $w(C_{\textrm{rep}}(S)) < p$ for every set $S \subseteq E$ of size at most $k$. We conclude that every point $(x',y')$ of $\mathcal{P}_I(r)$ satisfies $\sum_{c \in C_{\textrm{rep}}} |X_c| \cdot x'(c) \leq p - 1$.
\end{proof}

Lemmas~\ref{lemma:yes} and~\ref{lemma:no} show that if we are able to solve the resulting Max $k$-Edge Cover instance, then at each ellipsoid iteration we will be able to either round or cut, and thus, we will get the desired algorithm. So, the only thing remaining to do is solve the resulting Max $k$-Edge Cover instance. The following theorem shows that we can indeed solve Max $k$-Edge Cover in polynomial time.

\begin{theorem}\label{thm:cec}
There exists an algorithm that solves Max $k$-Edge Cover in polynomial time.
\end{theorem}

The above theorem can be proved by using standard reductions that may have appeared before in the literature. More precisely, we reduce Max $k$-Edge Cover to the Maximum Weight Perfect Matching problem, which is solvable in polynomial time. For completeness, we give a proof of Theorem~\ref{thm:cec} in Section~\ref{sec:matching}.

We now put everything together and give a formal proof of Theorem~\ref{thm:k-supplier-either-or}.

\begin{proof}[Proof of Theorem~\ref{thm:k-supplier-either-or}]
We use the ellipsoid method to certify emptiness of the polytope $\mathcal{P}_I(r)$. Given a candidate point $(x,y) \in \mathbb{R}^X \times \mathbb{R}^F$, we first check if $(x,y) \in \mathcal{P}(r) \supseteq \mathcal{P}_I(r)$. If $(x,y) \notin \mathcal{P}(r)$, then we can separate by using a violated inequality in the inequality description of $\mathcal{P}(r)$. Otherwise, we have $(x,y) \in \mathcal{P}(r)$. In this case, we run Algorithm~\ref{alg:RFiltering2} and obtain a set $C_{\textrm{rep}} \subseteq X$ and a subpartition $\{X_c\}_{c \in C_{\textrm{rep}}}$ of $X$. Based on these sets, we construct the correposponding Max $k$-Edge Cover instance, and use the algorithm of Theorem~\ref{thm:cec} to optimally solve it. If we obtain a solution of weight at least $p$, then by using Lemma~\ref{lemma:yes} we construct a feasible solution of our Robust Euclidean $k$-Supplier instance of radius at most $(1 + \sqrt{3})r$, which we return. Otherwise, if the obtained solution has weight strictly smaller than $p$, then by using Lemma~\ref{lemma:no}, we certify that $(x,y) \notin \mathcal{P}_I(r)$, and moreover, we get a separating hyperplane. It is easy to see that all generated hyperplanes have polynomial encoding length and thus, by Theorem 6.4.9 in \cite{grotschel2012geometric}, we are guaranteed that in polynomially many steps we will either obtain a feasible solution of radius at most $(1 + \sqrt{3})r$ or certify that $\mathcal{P}_I(r) = \emptyset$, which implies that there is no solution of radius $r$.
\end{proof}

\subsection{Reducing Max \texorpdfstring{$k$}{k}-Edge Cover to Maximum Weight Perfect Matching}\label{sec:matching}

\input{./maximum_matching_reduction.tex}

%% file: maximum_matching_reduction.tex
In this section, we discuss the proof of Theorem~\ref{thm:cec}. In order to design an algorithm that solves Max $k$-Edge Cover in polynomial time, we will reduce it to the Maximum Weight Perfect Matching problem, or, in short, MWPM, which admits polynomial-time algorithms. Our reduction has two steps. We first reduce Max $k$-Edge Cover to a Maximum Weight Matching problem with a cardinality constraint on the number of edges that are allowed to be picked, which we call Max Weight $k$-Matching, and then we reduce Max Weight $k$-Matching to MWPM. We note that both reductions use standard techiniques and ideas that may have appeared before in the literature; nevertheless, we were not able to find an explicit reference for these.

\paragraph{The first reduction.} We start with a Max $k$-Edge Cover instance $G=(V,E)$ with a weight function $w: V \to \mathbb{N}_{\geq 0}$, where  $k \in \mathbb{N}$. We first preprocess the instance. Since the given graph $G=(V, E)$ might contain parallel edges and self-loops, we do the following. For every pair of neighboring vertices $u$ and $v$, we remove all but one of the edges connecting them (we arbitrarily select which edges to remove). Moreover, for every vertex $u$ that has more than one self-loops, we again remove all but one of these self-loops, by arbitrarily choosing which ones to remove. Finally, we discard all the vertices that are isolated, i.e., all vertices whose degree is zero. It is easy to see that the resulting Max $k$-Edge Cover instance has the same optimal value as the original instance, and any solution of the new instance is feasible for the original one. By slightly abusing notation, let $G = (V, E)$ be the resulting instance. From now on, we work with this simplified Max $k$-Edge Cover instance, which does not contain parallel edges, contains at most one self-loop per vertex and every vertex has degree at least $1$.

Given the simplified Max $k$-Edge Cover instance $G=(V, E)$, we construct a graph $H_G = (U, L)$, where $U$ is the vertex set and $L$ is the edge set, as follows. We set $U = V \cup \overline{V}$, where $\overline{V}$ is a copy of $V$; the copy of a vertex $u \in V$ is denoted as $\bar{u} \in \bar{V}$. We also set $L = (E \setminus \mathrm{Loop}(E)) \cup \overline{E}$, where $\overline{E} \coloneqq \{(u, \bar{u}): \; u \in V\}$ and $\mathrm{Loop}(E) \subseteq E$ is the set of self-loops contained in $G$. Finally, we define a weight function $w': L \to \mathbb{R}_{\geq 0}$ on the edges of $L$ as follows:
\begin{itemize}
    \item $w'(u,v) = w(u) + w(v)$, for every $(u,v) \in E \setminus \mathrm{Loop}(E)$,
    \item $w'(u,\bar{u}) = w(u)$, for every $u \in V$.
\end{itemize}

We will now show that the Max $k$-Edge Cover instance has a feasible solution of weight at least $p$ if and only if the constructed Max Weight $k$-Matching instance has a matching consisting of at most $k$ edges whose weight is at least $p$.

\begin{lemma}\label{lemma:first-part-of-reduction}
Let $G = (V, E)$ and $w: V \to \mathbb{N}_{\geq 0}$ be a Max $k$-Edge Cover instance, where $k \in \mathbb{N}$, that does not contain any parallel edges, contains at most one self-loop per vertex and every vertex has degree at least $1$. Let $H_G=(U, L)$ and $w' : L \to \mathbb{R}_{\geq 0}$ be the corresponding Max Weight $k$-Matching instance, as described above. Then, the given Max $k$-Edge Cover instance has a feasible solution of weight at least $p \in \mathbb{R}_{\geq 0}$ if and only if the Max Weight $k$-Matching instance $H_G$ has a solution of weight at least $p$.
\end{lemma}
\begin{proof}
Suppose that the given Max $k$-Edge Cover instance has a solution $S \subseteq E$ that satisfies $|S| \leq k$ and $w(V(S)) \geq p$. Without loss of generality, we assume that $S$ is a minimal set that satisfies $w(V(S)) \geq p$. It is easy to see that $S$ can be decomposed into a maximal matching $M_0 \subseteq S$ and a set of edges $M_1 = S \setminus M_0$, where every edge in $M_1$ has one endpoint in $V(M_0)$ and one endpoint in $V \setminus V(M_0)$, and moreover $|V(M_1) \setminus V(M_0)| = |M_1|$; this latter equality follows from the fact that $S$ is minimal. We now define $M_2 \coloneqq \{(u,\bar{u}) \in L: u \in V(M_1) \setminus V(M_0)\}$. It is not hard to see that $M_0 \cup M_2$ is a matching in $H_G$, and moreover $|M_0 \cup M_2| = |M_0| + |M_1| \leq k$. The above discussion now implies that
\begin{equation*}
    w'(M_0 \cup M_2) = w'(M_0) + w'(M_2) = w(V(M_0)) + w\left(V(M_1) \setminus V(M_0) \right) = w(V(S)) \geq p.
\end{equation*}
We conclude that $M_0 \cup M_2$ is a matching in $H_G$ consisting of at most $k$ edges whose weight is at least $p$.

For the converse, we assume that $H_G$ has a matching $M$ of size at most $k$ whose weight is at least $p$, and we will show that the Max $k$-Edge Cover instance has a solution of weight at least $p$. For that, we write $M = M_0 \cup M_1$, where $M_0 \subseteq E \setminus \mathrm{Loop}(E)$ and $M_1 \subseteq \overline{E}$. We now define $S \coloneqq M_0 \cup S'$, where $S' \subseteq E$ is any set of $|M_1|$ edges of $G$ that satisfies $\{u \in V: (u, \bar{u}) \in M_1\} \subseteq V(S')$; clearly such a set $S'$ exists, since $G$ has no isolated vertices. It is not hard to see now that $|S| \leq |M| \leq k$, and moreover, $w(V(S)) \geq w'(M_0) + w'(M_1) \geq p$. This concludes the proof.
\end{proof}

\paragraph{The second reduction.} Here, we start with a Max Weight $k$-Matching instance and we will reduce it to a MWPM instance. To simplify notation, let $G = (V, E)$, $w: E \to \mathbb{R}_{\geq 0}$ be a given Max Weight $k$-Matching instance, where $k \in \mathbb{N}$. We construct the MWPM instance $(H = (U, L), w')$ where $w': L \to \mathbb{R}_{\geq 0}$ as follows. We set $U \coloneqq V \cup \overline{V} \cup K$, where $\overline{V}$ is a copy of $V$ and $K$ is a set of $2k$ new vertices. We also set $L \coloneqq E \cup \overline{E} \cup C_K \cup K_{\overline{V}}$, where
\begin{itemize}
    \item $\overline{E} = \{(v, \bar{v}): v \in V\}$, 
    \item $C_K = \{(c, c'): c \neq c' \in K\}$ and
    \item $K_{\overline{V}} = \{(\bar{v}, c): \bar{v}\in \overline{V}, c \in K\}$.
\end{itemize}
Finally, we define a weight function $w': L \rightarrow \mathbb{R}_{\geq 0}$ as follows:
\begin{itemize}
    \item $w'(e) = w(e)$ for every $e \in E$
    \item $w'(e) = 0$ for every $e \in L \setminus E$.
\end{itemize}

\begin{lemma}\label{lemma:second-part-of-reduction}
Let $G = (V, E)$, $w: E \to \mathbb{R}_{\geq 0}$ be a Max Weight $k$-Matching instance, where $k \in \mathbb{N}$, and let $H = (U, L)$ be the corresponding MWPM instance with $w' : L \to \mathbb{R}_{\geq 0}$, as defined above. Then, the given Max Weight $k$-Matching instance has a solution of weight at least $p \in \mathbb{R}_{\geq 0}$ if and only if the corresponding MWPM instance has a solution of weight at least $p$.
\end{lemma}
\begin{proof}
Suppose that $M \subseteq E$ is a matching that satisfies $|M| \leq k$ and $w(M) \geq p$. We now construct a perfect matching $M' \subseteq L$ of $H$ whose weight is equal to $w(M)$, as follows.
\begin{itemize}
    \item We add every edge of $M$.
    \item We match every vertex in $V \setminus V(M)$ with its copy in $\overline{V}$ (here, we again used the notation $V(M) \coloneqq \{v \in V: \exists e \in M \textrm{ s.t. }v \in e\}$).
    \item We match the remaining $2|M|$ vertices of $\overline{V}$ with any $2|M|$ vertices of $K$; since $|M| \leq k$, this is always feasible.
    \item We match the remaining $2k - 2|M|$ vertices of $K$ arbitrarily with each other, since $K$ forms a clique.
\end{itemize}
It is easy to see that the resulting set $M'$ is indeed a perfect matching in $H$, whose weight is equal to $w'(M') = w(M) \geq p$.

Conversely, suppose that $M'$ is a perfect matching in $H$ that satisfies $w'(M') \geq p$. Clearly, we have $w'(M') = w'(M' \cap E) = w(M' \cap E)$. Let $M \coloneqq M' \cap E$. The set $M$ is a matching in $G$ of weight at least $p$. Thus, the only thing remaining to prove is that $|M| \leq k$. For that, we will show that any perfect matching in $H$ must pick at most $k$ edges from the set $E$. Suppose otherwise, i.e., suppose that $|M| > k$. For this to happen, we must have $|V| > 2k$. We now observe that the vertices in $\overline{V}$ can only be matched with vertices in $V$ or $K$. Since there are strictly fewer than $|V| - 2k$ vertices of $V$ that are not matched via $M$, it means that in total, there are strictly fewer than $|V|$ vertices among $V$ and $K$ that can be matched with vertices in $\overline{V}$. Since $|\overline{V}| = |V|$, this implies that no perfect matching exists, which is a contradiction. We conclude that we must necessarily have $|M| \leq k$.
\end{proof}

\paragraph{The complete reduction.} We now put everything together and prove Theorem~\ref{thm:cec}.
\begin{proof}[Proof of Theorem~\ref{thm:cec}]
We start with a Max $k$-Edge Cover instance, and we preprocess it so that there are no parallel edges, there is at most one self-loop per vertex, and every vertex has degree at least 1. We now use Lemma~\ref{lemma:first-part-of-reduction} and obtain a Max Weight $k$-Matching instance in polynomial time. Applying Lemma~\ref{lemma:second-part-of-reduction}, we obtain our final MWPM instance. Due to the guarantees of the intermediate lemmas, an optimal solution for the MWPM instance translates into an optimal solution of the original Max $k$-Edge Cover instance. Thus, we can now use some of the well-known algorithms for solving MWPM in order to obtain an optimal solution for our Max $k$-Edge Cover instance. This concludes the proof.
\end{proof}